\documentclass{llncs}

\newcommand{\Oh}[1]
	{\ensuremath{\mathcal{O}\!\left({#1}\right)}}
\newcommand{\Primes}
	{\ensuremath{\mathsf{Primes}}}
\newcommand{\true}
	{\ensuremath{\mathsf{true}}}
\newcommand{\false}
	{\ensuremath{\mathsf{false}}}

\begin{document}

\title{String Cadences}
\author{Amihood Amir\inst{1} \fnmsep \inst{2} \fnmsep \thanks{Partly supported by ISF grant 571/14.} \and
Alberto Apostolico\inst{3} \and \\
Travis Gagie\inst{4} \and
Gad M. Landau\inst{5} \fnmsep \inst{6} \fnmsep $^\star$}
\institute{Department of Computer Science, Bar-Ilan University, Ramat-Gan 52900, Israel. \email{E-mail: amir@cs.biu.ac.il}
\and Department of Computer Science, Johns Hopkins University, Baltimore, MD 21218.
\and School of Computational Science and Engineering, College of Computing, Georgia Institute of Technology, Klaus Advanced Computing Building, 266 Ferst Drive, Atlanta, GA 30332-0765. \email{Email: axa@cc.gatech.edu}
\and School of Computer Science and Telecommunications, Diego Portales University, Santiago, Chile. \email{E-mail: travis.gagie@mail.udp.cl}
\and Department of Computer Science, University of Haifa, Mount Carmel, Haifa 31905, Israel. \email{Email: landau@cs.haifa.ac.il}
\and Department of Computer Science and Engineering, NYU Polytechnic School of Engineering, 2 MetroTech Center, Brooklyn, NY 11201.}
\maketitle

\begin{abstract}
We say a string has a cadence if a certain character is repeated at regular intervals, possibly with intervening occurrences of that character.  We call the cadence anchored if the first interval must be the same length as the others.  We give a sub-quadratic algorithm for determining whether a string has any cadence consisting of at least three occurrences of a character, and a nearly linear algorithm for finding all anchored cadences.
\end{abstract}

\section{Introduction}
\label{sec:introduction}

Finding interesting patterns in strings is an important problem in many fields, with an extensive literature (see, e.g.,~\cite{Apostolico2003} and references therein).  One might therefore expect all obvious kinds to have already been considered; however, as far as we are aware, no one has previously investigated how best to determine whether a string contains a certain character repeated at regular intervals, possibly with intervening occurrences of that character.  To initiate the study of this natural problem, we introduce the following notions:

\begin{definition}
\label{def:general}
A {\em cadence} of a string \(S [1..n]\) is a pair \((i, d)\) of natural numbers such that \(i \leq d \leq n\) and \(S [j] = S[i]\) for every $j$ between 1 and $n$ with \(j \equiv i \bmod d\).  A cadence \((i, d)\) with \(\lfloor (n - i) / d \rfloor + 1 = k\) is a {\em $k$-cadence}.
\end{definition}

\begin{definition}
\label{def:anchored}
An {\em anchored cadence} of a string \(S [1..n]\) is a natural number $i$ such that \((i, i)\) is a cadence of $S$: i.e., such that \(S [j] = S [i]\) for every $j$ between 1 and $n$ with \(j \equiv 0 \bmod i\).
\end{definition}

Informally, if the first interval must be the same length as the others then the cadence is anchored.  For example, if \(S = \mathtt{ALABARALAALABARDA}\) then \((3, 7)\) is a 3-cadence, because \(S [3] = S [10] = S [17] = \mathtt{A}\), and 7 is an anchored cadence, because \(S [7] = S [14] = \mathtt{A}\).  

If \(d \geq \max (i, n - i + 1)\) then \((i, d)\) is trivially a 1-cadence; otherwise, we can check that \((i, d)\) is a cadence by comparing \(S [i], S [2 i], \ldots, S [i + \lfloor (n - i) / d \rfloor d]\).  Therefore, we can find all cadences of $S$ in time at most proportional to
\[\sum_{i = 1}^n \sum_{d = i}^n n / d = n \sum_{d = 1}^n \sum_{i = 1}^d 1 / d = n^2\,.\]
Since there can still be \(\Theta (n^2)\) $k$-cadences for \(k \geq 2\), this bound is worst-case optimal.  In Section~\ref{sec:3-cadences}, however, we give a sub-quadratic algorithm for determining whether a string has a 3-cadence.  In Section~\ref{sec:anchored} we give an $\Oh{n \log \log n}$ time algorithm for finding all anchored cadences, of which there can be at most $n$.

We leave as open problems finding an output-sensitive algorithm for reporting all the $k$-cadences for \(k \geq 2\), or a subquadratic algorithm for determining whether there is a $k$-cadence for a given \(k \geq 4\), or a linear algorithm for finding all anchored cadences.  We are also curious about how to define properly and find efficiently approximate cadences, and whether there exists, e.g., a subquadratic-space data structure that, given the endpoints of a substring, quickly reports all the cadences of that substring.

\section{Detecting 3-Cadences}
\label{sec:3-cadences}

A string has a 2-cadence if, for some character $a$, the positions $i$ and $j$ of the leftmost and rightmost occurrences of $a$ satisfy \(2 i \leq j\) and \(2 j - i > n\).  We can easily check this in $\Oh{n \log n}$ time.  In this section we show that when the string is binary, we can also check if it has a 3-cadence in $\Oh{n \log n}$ time.  It follows that we still need only $\Oh{n \log n}$ time when the string is over any constant-size alphabet, and $\Oh{n^{3 / 2} \log^{1 / 2} n}$ time in general.

Specifically, we show how to convert a string \(S [1..n] \in \{0, 1\}^n\) into a set $W$ of $n$ integer weights in \([- 2 n, \ldots, 2 n]\) such that $S$ has a 3-cadence \((i, d)\) with \(S [i] = 1\) if and only if $W$ is in {\sc 3Sum} (i.e., three of its weights sum to 0), which we can check in $\Oh{n \log n}$ time via the Fast-Fourier Transform (see, e.g.,~\cite[Exercise 30.1--7]{CormenLeisersonRivestStein2001}).  Since our reduction is essentially reversible, we suspect that improving this $\Oh{n \log n}$ bound will be challenging.

Without loss of generality, assume we are interested only in detecting 3-cadences \((i, d)\) with \(S [i] = 1\); we can detect 3-cadences \((i, d)\) with \(i = 0\) symmetrically.  Let
\begin{eqnarray*}
L_1 & = & \{j\,:\,S [j] = 1, j \leq n / 3\}\\
L_2 & = & \{j\,:\,S [j] = 1, n / 3 < j \leq 2 n / 3\}\\
L_3 & = & \{j\,:\,S [j] = 1, 2 n / 3 < j\}\,.
\end{eqnarray*}
By definition, for any 3-cadence \((i, d)\) we have \(i \in L_1\), \(i + d \in L_2\) and \(i + 2 d \in L_3\).  Therefore, there is a 3-cadence if and only if the average of some element in $L_1$ and some element in $L_3$ is an element in $L_2$.  It follows that $S$ has a 3-cadence if and only if
\[L_1 \cup L_3 \cup \{- 2 j\,:\,j \in L_2\} \in \mbox{\sc 3Sum}\,.\]

If $S$ is over any constant-size alphabet then we can create a binary string of length $n$ for each character in the alphabet, with the 1s marking the occurrences of that character, and apply our reduction to each one in $\Oh{n \log n}$ total time.  If $S$ is over an arbitrary alphabet then we can perform this partitioning and detect 3-cadences in the binary string for each character $a$ using time at most proportional to \(\min \left( n_a^2, n \log n \right)\), where $n_a$ is the number of occurrences of $a$ in $S$.  This takes a total of at most $\Oh{n^{3 / 2} \log^{1 / 2} n}$ time.

\begin{theorem}
\label{thm:3-cadences}
We can determine whether a string of length $n$ has a 3-cadence in $\Oh{n^{3/2} \log^{1/2} n}$ time.  If the alphabet has constant size then we use $\Oh{n \log n}$ time.
\end{theorem}

To see that our reduction is essentially reversible, suppose we want to determine whether a set $W$ of $n$ integer weights in \([-n, \ldots, n]\) is in {\sc 3Sum}.  Without loss of generality, we need check only whether there exist two positive weights and one negative weight that sum to 0: the case when there are two negative weights and one positive weight that sum to 0 is symmetric.  For each positive weight $w$ we set \(S [2 w]\) to 1; for each negative weight $w$ we set \(S [- w]\) to 1; and for each remaining \(i \leq n\) we set \(S [i]\) to 0 if \(i \leq n / 2\) and to 2 otherwise.  There are two positive weights $w_1$ and $w_3$ and one negative weight $w_2$ in $W$ with \(w_1 + w_2 + w_3 = 0\) if and only if \(- w_2\) is the average of \(2 w_1\) and \(2 w_3\), in which case $S$ has a 3-cadence.

\section{Finding Anchored Cadences}
\label{sec:anchored}

We can check whether \(i \leq n\) is an anchored cadence of \(S [1..n]\) by comparing \(S [i], S [2 i], \ldots, S [\lfloor n / i \rfloor i]\), which takes $\Oh{n / i}$ time.  Since
\[\sum_{i \leq n} 1 / i = \Oh{\log n}\,,\]
obviously we can find all anchored cadences in $\Oh{n \log n}$ time.  In this section we use the following lemma to improve this bound to $\Oh{n \log \log n}$.

\begin{lemma}
\label{lem:anchored}
If a natural number \(i \leq n\) is not an anchored cadence of a string \(S [1..n]\), then for some prime $p$ with \(p\,i \leq n\), either \(S [i] \neq S [p\,i]\) or \(p\,i\) is not an anchored cadence of $S$.
\end{lemma}

\begin{proof}
Let \(p\,i\) be the smallest multiple of $i$ greater than $i$ itself such that either \(S [i] \neq S [p\,i]\) or \(p\,i\) is not an anchored cadence; \(p\,i\) must exist or $i$ would be an anchored cadence.  To see why $p$ must be prime, assume it has a prime factor \(r < p\): then \(S [i] = S [p\,i / r]\) and \(p\,i / r\) is an anchored cadence, by our choice of $p$, meaning \(S [i] = S [p\,i / r] = S [p\,i]\) and \(S [p\,i]\) is an anchored cadence, contradicting our choice of $p$. \qed
\end{proof}

We start by computing the sorted set \(\Primes (n)\) of primes between 1 and $n$, which takes $\Oh{n / \log \log n}$ time~\cite{Pritchard1994}.  We then build a Boolean array \(B [1..n]\) and set \(B [i]\) to $\true$ for \(\lceil n / 2 \rceil \leq i \leq n\).  For $i$ from \(\lceil n / 2 \rceil - 1\) down to 1, we check whether \(S [i] = S [p\,i]\) and \(B [p\,i] = \true\) for each prime $p$ with \(p\,i \leq n\), in increasing order; if so, we set \(B [i]\) to $\true$ and otherwise we set it to $\false$.  By Lemma~\ref{lem:anchored}, $i$ is an anchored cadence if and only if we eventually set \(B [i]\) to $\true$.  We check each cell \(B [j]\) at most once for each of the distinct prime factors of $j$.  Since each prime \(p \leq n\) divides \(\lfloor n / p \rfloor\) numbers between 1 and $n$ and
\[\sum_{p \in \Primes (n)} 1 / p = \ln \ln n + \Oh{1}\]
(see, e.g.,~\cite{BachShallit1996}), we use $\Oh{n \log \log n}$ time.

If $m$ is the smallest anchored cadence of $S$ then, for \(i < m / 2\), we set \(B [i]\) to $\false$ immediately after checking \(B [2 i]\).  Therefore we check each cell \(B [j]\) at most once for each prime factor $p$ of $j$ such that \(p = 2\) or \(j / p \geq m / 2\).  For each prime \(p \leq 2 n / m\) there are \(n / p - m / 2 + \Oh{1}\) choices of $j$ between 1 and $n$ such that $p$ divides $j$ and \(j / p \geq m / 2\); there is no such choice of $j$ for any larger prime.  It follows that we use time $\Oh{n \log \log (n / m)}$ time.  Moreover, if $S$ is chosen uniformly at random from over a non-unary alphabet then in the expected case, for \(i \leq n\), we check at most 2 primes before finding one $p$ such that \(S [i] \neq S [p\,i]\).  In the average case, therefore, we use $\Oh{n}$ total time.

\begin{theorem}
\label{thm:anchored}
We can find all the anchored cadences in a string of length $n$ in $\Oh{n \log \log (n / m)}$ time, where $m$ is the smallest cadence.  If the alphabet is non-unary then we use $\Oh{n}$ time on average.
\end{theorem}

\section*{Acknowledgments}

Many thanks to the organizers and other participants of the AxA workshop.

\end{document}